\documentclass[12pt,a4paper]{article}
\usepackage[utf8]{inputenc}
\usepackage[T1]{fontenc}
\usepackage{lmodern}
\usepackage{natbib}
\usepackage{amsmath,amssymb,amsthm}
\usepackage{booktabs}
\usepackage{graphicx}
\usepackage{geometry}
\usepackage{xcolor}
\usepackage[colorlinks,linkcolor=blue,citecolor=blue!50!black,urlcolor=blue]{hyperref}
\usepackage{setspace}
\newtheorem{assumption}{Assumption}
\newtheorem{proposition}{Proposition}
\newtheorem{remark}{Remark}

\newcommand{\cum}{\operatorname{cum}}

\title{Testing the Exclusion Restriction with a Single Instrument:\\ Identification through Fourth-Order Cumulants under Latent Confounding}
\author{Fernando Delbianco\thanks{Department of Economics, Universidad Nacional del Sur, and INMABB--CONICET, Bah\'ia Blanca, Argentina. \href{mailto:fernando.delbianco@uns.edu.ar}{fernando.delbianco@uns.edu.ar}. ORCID 0000-0002-1560-2587.}}
\date{}

\begin{document}
\maketitle

\begin{abstract}
\noindent With a single instrument the exclusion restriction is untestable from second moments, and adding the instrument to the structural equation fails because the regressor is endogenous. When the endogeneity comes from a latent common cause with non-Gaussian components, the restriction becomes generically testable from the fourth-order cumulants of the first-stage and 2SLS residuals: a closed-form scalar $\theta$ is zero under exclusion and non-zero otherwise, while the direct effect itself is identified only up to two points. A companion statistic measures identification strength. Simulations show correct size under confounding, where the naive OLS test always rejects. In Card's (1995) data the test is uninformative, and says so.\\[4pt]
\end{abstract}

\section{Introduction}\label{sec:intro}

Consider the textbook linear IV model with one instrument $Z$, one endogenous regressor $X$ and outcome $Y$,
\begin{equation}\label{eq:model}
X = \pi Z + v, \qquad Y = \beta X + \gamma Z + u, \qquad Z \perp\!\!\!\perp (u,v),
\end{equation}
where $u$ and $v$ are dependent (that is what makes $X$ endogenous) and $\gamma=0$ is the exclusion restriction. With only one instrument the model is just identified, and the exclusion restriction is not testable from the covariance matrix of $(Z,X,Y)$: the reduced form delivers $\pi$ and $\rho=\beta\pi+\gamma$, two numbers for three structural parameters. Regress $Y$ on $(X,Z)$ and test the coefficient of $Z$ does not work, because $X$ is correlated with the error and the coefficient of $Z$ inherits a bias of $-\pi\,\mathrm{Cov}(u,v)/\mathrm{Var}(v)$ even when $\gamma=0$ \citep[Problem 5.5]{wooldridge2010}. This is the sense in which identification cannot be tested with a single instrument.

This note shows that the conclusion rests on using second moments only. If the dependence between $u$ and $v$ comes from a latent common cause and the unobservables are non-Gaussian, the fourth-order cumulants of two residuals (i.e. the first-stage and the 2SLS residuals) satisfy a restriction under $\gamma=0$ that fails, generically, under $\gamma\neq 0$. The statistic is a closed-form function of sample cumulants, and its distribution is bootstrapped.

The idea joins two literatures that rarely meet. In econometrics, higher-order moments identify what second moments cannot: \citet{reiersol1950} and \citet{geary1942} for errors in variables, \citet{cragg1997}, \citet{dagenais1997}, \citet{lewbel1997} and \citet{erickson2002} for higher-moment instruments and cumulant GMM, and \citet{rigobon2003}, \citet{lanne2017}, \citet{gouriéroux2017}, \citet{guay2021} and \citet{lanne2026} for structural VARs identified through heteroskedasticity or non-Gaussianity, with \citet{arias2021} on proxy validity. In machine learning, the same mathematics (identifiability of overcomplete independent component analysis (ICA) \citep{eriksson2004}) underlies non-Gaussian causal discovery \citep{shimizu2006,shimizu2011}, its extension to latent confounders \citep{hoyer2008,salehkaleybar2020,wang2023} and the search for valid instruments \citep{silva2017}. The contribution here is a test that is directly about the exclusion restriction of one instrument under the confounding that motivates IV. Applying the confounder-free LiNGAM model of \citet{shimizu2006} to $(Z,X,Y)$ would not do, because it assumes independent structural errors, so OLS is consistent and the test is valid only when it is not needed (Section~\ref{sec:mc}).

The framework is linear with homogeneous effects. It complements tests of instrument validity under heterogeneous effects \citep{kitagawa2015,huber2015,mourifie2017}, which rely on inequality restrictions rather than non-Gaussianity, and differs from approaches that relax the restriction \citep{conley2012,nevo2012} or diagnose first-stage heterogeneity \citep{dieterle2016}. The price of testability is a maintained common-factor structure and non-zero fourth cumulants, and the latter is measurable in the data.

\section{Model and identification}\label{sec:theory}

\begin{assumption}\label{a:model}
$(Z,X,Y)$ is generated by \eqref{eq:model} with
\begin{equation}\label{eq:factor}
v = \lambda_1 H + e_X, \qquad u = \lambda_2 H + e_Y,
\end{equation}
where $H$, $e_X$, $e_Y$ are mutually independent, mean zero, with finite eighth moments, and independent of $Z$. The instrument is relevant, $\pi\neq 0$, and it is confounded with the outcome through $H$: $\lambda_1\lambda_2\neq 0$.
\end{assumption}

\begin{assumption}\label{a:ng}
$\kappa_4(H)\neq 0$ and $\kappa_4(e_X)\neq 0$, where $\kappa_4(\cdot)$ denotes the fourth cumulant (excess kurtosis times squared variance).
\end{assumption}

Assumption \ref{a:model} is the single-factor IV problem: $X$ is endogenous because an unobservable $H$ (e.g. ability) enters both equations. Nothing is assumed about the distribution of $Z$, which may be binary, and $e_Y$ may be Gaussian. Assumption \ref{a:ng} restricts only the two unobservables that load on $X$.

Let $c=\gamma/\pi$. The reduced-form coefficients are $\pi$ and $\rho=\beta\pi+\gamma$, and the population 2SLS coefficient is $\beta_{IV}=\rho/\pi=\beta+c$. Define the two population residuals
\begin{equation}\label{eq:resid}
v = X-\pi Z, \qquad w = Y-\beta_{IV}X = u - c\,v .
\end{equation}
The instrument drops out of $w$ whether or not $\gamma=0$. This is the algebraic form of untestability, because no test of dependence between $Z$ and the 2SLS residual has power. The information is in the joint distribution of $(v,w)$. Writing $\tilde H=\lambda_1 H$, \eqref{eq:factor} and \eqref{eq:resid} give
\begin{equation}\label{eq:mix}
v = \tilde H + e_X, \qquad w = m\,\tilde H + d\,e_X + e_Y, \qquad m = \lambda_2/\lambda_1 - c,\quad d = -c .
\end{equation}
Under exclusion, $d=0$. In other words, an independent non-Gaussian component ($e_X$) enters $v$ but not $w$. Under a violation every component of $v$ enters $w$. Equation \eqref{eq:mix} is a $2\times 3$ overcomplete ICA whose mixing directions are identified when at most one source is Gaussian \citep{eriksson2004}. The proposition makes this operational.

Let $s_r=\cum(\underbrace{v,\dots,v}_{4-r},\underbrace{w,\dots,w}_{r})$, $r=0,1,2,3$, be the fourth-order cross cumulants of $(v,w)$, and let $A=\kappa_4(\tilde H)$, $B=\kappa_4(e_X)$.

\begin{proposition}\label{prop:main}
Under Assumption \ref{a:model},
\begin{enumerate}
\item[(i)] $s_r = m^r A + d^r B$ for $r=0,1,2,3$; hence $s_{r+2}=(m+d)s_{r+1}-md\,s_r$ for $r=0,1$.
\item[(ii)] $s_0s_2-s_1^2 = AB\,(m-d)^2 = AB\,(\lambda_2/\lambda_1)^2$ and $s_1s_3-s_2^2 = AB\,md\,(m-d)^2$.
\item[(iii)] Under Assumption \ref{a:ng} in addition, $D\equiv s_0s_2-s_1^2\neq 0$ and
\begin{equation}\label{eq:theta}
\theta \equiv \frac{s_1s_3-s_2^2}{s_0s_2-s_1^2} = md = c\left(c-\frac{\lambda_2}{\lambda_1}\right).
\end{equation}
Consequently $\gamma=0$ implies $\theta=0$, and $\theta=0$ implies $\gamma\in\{0,\ \pi\lambda_2/\lambda_1\}$.
\item[(iv)] The unordered pair $\{m,d\}$ is identified as the roots of $x^2-(m+d)x+md$, with $m+d=(s_0s_3-s_1s_2)/D$. Hence $\gamma$ is set-identified: $\gamma\in\{\pi c^{*}: c^{*}\in\{-m,-d\}\}$, a two-point set.
\end{enumerate}
\end{proposition}

\begin{proof}
(i) Cumulants are multilinear and cross cumulants of independent variables vanish, so by \eqref{eq:mix} $\cum(v^{(4-r)},w^{(r)})=1^{4-r}m^{r}\kappa_4(\tilde H)+1^{4-r}d^{r}\kappa_4(e_X)+0^{4-r}1^{r}\kappa_4(e_Y)$, and the last term is zero for $r\le 3$. A sequence $m^rA+d^rB$ satisfies the recurrence with characteristic roots $m$ and $d$. (ii) Direct substitution: $s_0s_2-s_1^2=(A+B)(m^2A+d^2B)-(mA+dB)^2=AB(m-d)^2$ and $s_1s_3-s_2^2=(mA+dB)(m^3A+d^3B)-(m^2A+d^2B)^2=AB\,md(m-d)^2$; $m-d=\lambda_2/\lambda_1$ from \eqref{eq:mix}. (iii) $A,B\neq 0$ by Assumption \ref{a:ng} and $\lambda_2/\lambda_1\neq 0$ by Assumption \ref{a:model}, so $D\neq 0$ and the ratio equals $md$; substituting $m,d$ from \eqref{eq:mix} gives \eqref{eq:theta}, whose zeros are $c=0$ and $c=\lambda_2/\lambda_1$. (iv) The two-by-two linear system formed by the recurrence at $r=0,1$ has determinant $-D\neq 0$; solving it gives $m+d$ and $md$, hence the roots. Exchanging $(\tilde H,e_X)$ maps $(m,d)$ to $(d,m)$ without changing the distribution of $(v,w)$, so the roots cannot be labelled.
\end{proof}

\begin{remark}[What is and is not identified]
The restriction is testable although $\gamma$ is not point-identified: $\gamma=0$ says that zero is a root of the recurrence, a labelling-free property, whereas labelling the roots is the ICA ambiguity. A violation is masked only at $\gamma=\pi\lambda_2/\lambda_1$, where the instrument's direct effect per unit of first-stage effect equals the confounder's ratio; power vanishes there (Figure~\ref{fig:power}).
\end{remark}

\begin{remark}[Identification strength]
$D=AB(\lambda_2/\lambda_1)^2$ collects the three ingredients the test needs: non-Gaussianity of the confounder and of the idiosyncratic first-stage error, and confounding itself. If any is absent, $D=0$ and $\theta$ is undefined (with $\lambda_2=0$ there is no endogeneity and the naive OLS test is valid). A bootstrap $t$-ratio for $\hat D$ therefore plays the role the first-stage $F$ plays for relevance, and should be reported with $\hat\theta$.
\end{remark}

\begin{remark}[Extensions]
With $k$ latent factors the sequence has rank $k+1$ and the argument applies to the $(k+1)\times(k+1)$ Hankel matrix of cumulants up to order $2k+2$. Under skewness the third-order sequence obeys the same recurrence and adds the restriction $t_2s_1-t_1s_2=0$. Covariates are partialled out of $(Z,X,Y)$ first, leaving \eqref{eq:mix} unchanged.
\end{remark}

\section{Test statistic}\label{sec:test}

Let $\hat\pi$ and $\hat\beta_{IV}$ be the OLS first-stage and 2SLS coefficients (after partialling out covariates, if any), $\hat v=X-\hat\pi Z$ and $\hat w=Y-\hat\beta_{IV}X$ the demeaned residuals, standardised to unit variance, and $\hat s_r$ the plug-in sample cumulants, e.g.\ $\hat s_2=\overline{\hat v^2\hat w^2}-1-2\,\overline{\hat v\hat w}^{\,2}$. The statistic is $\hat\theta=(\hat s_1\hat s_3-\hat s_2^2)/(\hat s_0\hat s_2-\hat s_1^2)$. A smooth function of sample moments up to order four, with $O_p(n^{-1/2})$ error from $\hat\pi,\hat\beta_{IV}$, $\hat\theta$ is $\sqrt n$-consistent and asymptotically normal under Assumptions \ref{a:model}--\ref{a:ng}, and the nonparametric bootstrap is consistent for its distribution \citep{hall1992}. The test rejects when $|\hat\theta|/\widehat{se}_{boot}$ exceeds the normal critical value ($B=199$); the same resamples give $\hat D/\widehat{se}_{boot}$.

\section{Monte Carlo evidence}\label{sec:mc}

Data follow \eqref{eq:model}--\eqref{eq:factor} with $\pi=0.7$, $\beta=0.5$, $\lambda_1=\lambda_2=1$, a Gaussian instrument $Z\sim N(0,1)$, and $(H,e_X,e_Y)$ i.i.d.\ unit-variance Laplace ($\kappa_4=3$), centred exponential ($\kappa_4=6$) or symmetric bimodal Gaussian mixture ($\kappa_4=-1.28$); the masking point is $\gamma=0.7$. Table~\ref{tab:mc} reports 5\% rejection frequencies (500 replications) for the cumulant test and for the naive OLS $t$-test of $Z$ in the regression of $Y$ on $(X,Z)$.

\begin{table}[htbp]
\centering\footnotesize\setlength{\tabcolsep}{4pt}
\caption{Rejection frequencies at the 5\% level, 500 replications, $B=199$}
\label{tab:mc}
\begin{tabular}{lcccccccccc}
\toprule
 & & \multicolumn{4}{c}{Cumulant test, $H_0:\theta=0$} & & \multicolumn{4}{c}{Naive OLS $t$-test on $Z$} \\
\cmidrule{3-6}\cmidrule{8-11}
 Errors & $\gamma\ \backslash\ n$ & 500 & 1000 & 2000 & 4000 & & 500 & 1000 & 2000 & 4000 \\
\midrule
Laplace ($\kappa_4=3$) & 0.0 & 0.000 & 0.006 & 0.024 & 0.056 & & 0.998 & 1.000 & 1.000 & 1.000 \\
 & 0.1 & 0.004 & 0.024 & 0.158 & 0.314 & & 0.986 & 1.000 & 1.000 & 1.000 \\
 & 0.2 & 0.002 & 0.088 & 0.498 & 0.800 & & 0.664 & 0.940 & 0.994 & 1.000 \\
 & 0.3 & 0.002 & 0.164 & 0.794 & 0.984 & & 0.172 & 0.208 & 0.374 & 0.630 \\
\addlinespace
Centred exponential ($\kappa_4=6$) & 0.0 & 0.000 & 0.006 & 0.044 & 0.052 & & 1.000 & 1.000 & 1.000 & 1.000 \\
 & 0.1 & 0.002 & 0.060 & 0.264 & 0.480 & & 0.968 & 1.000 & 1.000 & 1.000 \\
 & 0.2 & 0.012 & 0.214 & 0.678 & 0.898 & & 0.662 & 0.912 & 0.994 & 1.000 \\
 & 0.3 & 0.024 & 0.324 & 0.902 & 0.986 & & 0.158 & 0.218 & 0.412 & 0.646 \\
\addlinespace
Bimodal mixture ($\kappa_4=-1.28$) & 0.0 & 0.012 & 0.022 & 0.036 & 0.056 & & 1.000 & 1.000 & 1.000 & 1.000 \\
 & 0.1 & 0.050 & 0.244 & 0.352 & 0.646 & & 0.980 & 1.000 & 1.000 & 1.000 \\
 & 0.2 & 0.196 & 0.666 & 0.924 & 1.000 & & 0.730 & 0.946 & 1.000 & 1.000 \\
 & 0.3 & 0.448 & 0.922 & 1.000 & 1.000 & & 0.124 & 0.204 & 0.340 & 0.618 \\
\bottomrule
\multicolumn{11}{l}{\footnotesize DGP: $\pi=0.7$, $\beta=0.5$, $\lambda_1=\lambda_2=1$, $Z\sim N(0,1)$; $\gamma=0$ is the null. Nominal level 5\%.}\\
\end{tabular}

\end{table}

First, under confounding the naive OLS test rejects a true null, as \citet{wooldridge2010} predicts, and its rejection rate then \emph{falls} as $\gamma$ grows (to about 20\% at $\gamma=0.3$, $n=1000$) because bias and true effect offset. The confounder-free LiNGAM estimator is no better: with $\gamma=0$, $n=5000$ and $t_5$ errors its direct-effect estimate converges to $-0.22$ ($\lambda_1=\lambda_2=1$) or $+0.35$ ($\lambda_2=-1$), then also reversing the order of $X$ and $Y$. Second, the cumulant test is conservative below $n=2000$ (null rejection rates of 0--4\%) and at the nominal level by $n=4000$. Third, power grows with $n$ and with the precision of sample kurtosis: with the bounded-support mixture a violation of $\gamma=0.2$ is detected in two thirds of samples at $n=1000$ and 92\% at $n=2000$; with heavier-tailed Laplace and exponential errors similar power needs $n=4000$ and $n=2000$--$4000$. Fourth-cumulant identification is weaker than second-moment identification; it is a tool for survey-sized samples. Figure~\ref{fig:power} shows the dip at the masking point.

\begin{figure}[htbp]
\centering
\includegraphics[width=.8\linewidth]{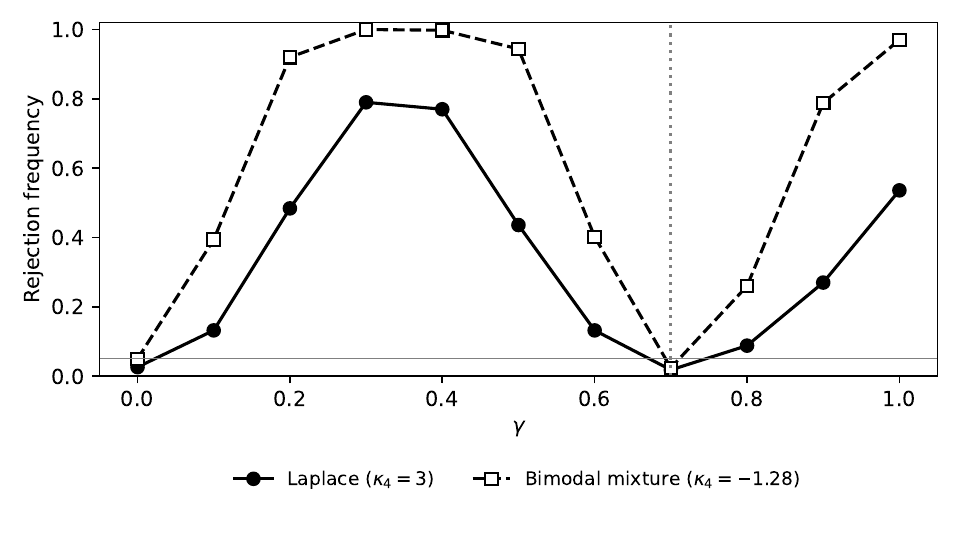}
\caption{Rejection frequency against $\gamma$ at $n=2000$ (500 replications). The vertical line marks the masking point $\gamma=\pi\lambda_2/\lambda_1=0.7$.}
\label{fig:power}
\end{figure}

\section{Application: Card (1995)}\label{sec:card}

\citet{card1995} instruments schooling with an indicator of growing up near a four-year college (NLSYM, $n=3010$; data from \citealp{wooldridge2010}); ability is the canonical $H$. Table~\ref{tab:card} reports IV output and cumulant statistics without covariates and after partialling out Card's controls (experience and its square, race, South, SMSA, SMSA in 1966, region dummies).

\begin{table}[htbp]
\centering\small\setlength{\tabcolsep}{4pt}
\caption{Card (1995) data: exclusion-restriction test for \texttt{nearc4}}
\label{tab:card}
\begin{tabular}{lcc}
\toprule
 & No covariates & Card's covariates \\
\midrule
Observations & 3010 & 3010 \\
First stage $\hat\pi$ & 0.829 & 0.320 \\
First-stage $F$ & 63.9 & 13.3 \\
$\hat\beta_{OLS}$ & 0.052 & 0.075 \\
$\hat\beta_{IV}$ (s.e.) & 0.188 (0.026) & 0.132 (0.055) \\
Excess kurtosis of $\hat v$, $\hat w$ & 0.24, 0.13 & -0.05, 0.48 \\
Jarque--Bera $p$-value, $\hat v$, $\hat w$ & 0.000, 0.032 & 0.000, 0.000 \\
$\hat D$ ($t$-ratio) & 0.0033 (0.36) & 0.0039 (0.45) \\
$\hat\theta$ (s.e.) & -0.26 (12.7) & 1.94 (51.4) \\
$p$-value, $H_0:\theta=0$ & 0.984 & 0.970 \\
Naive OLS coefficient on $Z$ ($p$-value) & 0.115 (0.000) & 0.018 (0.279) \\
\bottomrule
\multicolumn{3}{l}{\footnotesize Bootstrap s.e., $B=1999$. Covariates partialled out of $(Z,X,Y)$ first.}\\
\end{tabular}

\end{table}

The instrument is strong and the residuals reject normality, but their excess kurtosis is small ($-0.05$ to $0.5$), so $AB$ in Proposition~\ref{prop:main}(ii) is near zero: $\hat D$ has a $t$-ratio of $0.3$--$0.5$ and $\hat\theta$ a standard error around $13$ or higher. The data are consistent with exclusion and equally with sizeable violations: this sample is uninformative about the restriction, as a first-stage $F$ of two would be about $\beta$. The naive OLS coefficient on \texttt{nearc4} without controls, $0.115$ ($p<0.001$), falsely rejects the null, reflecting ability bias rather than a structural violation.

\section{Concluding remarks}

Between the untestable second-moment case and the trivial independent-errors case lies the one applied researchers face: a regressor confounded by a latent factor. There, fourth-order cumulants of the first-stage and 2SLS residuals reveal whether zero is a root of a two-term recurrence, exactly the exclusion restriction, while the direct effect stays identified only up to two points. The test has correct size where the naive regression fails, needs no assumption on the instrument's distribution, and comes with a diagnostic that says when the data cannot answer. Its limits are those of higher-moment identification, i.e. power needs kurtosis and large samples, and one masking point exists. Several factors, third-cumulant restrictions and weak-identification-robust inference are natural extensions.

\bigskip
\noindent\textbf{Data and code}: the Python implementation and replication scripts for all tables and figures are available at the Zenodo repository of the paper: \texttt{https://bit.ly/3UNg1OX}.

\bibliographystyle{apalike}

\begin{thebibliography}{99}

\bibitem[Arias et al.(2021)]{arias2021} Arias, J.~E., Rubio-Ram\'irez, J.~F., \& Waggoner, D.~F. (2021). Inference in Bayesian proxy-SVARs. \textit{Journal of Econometrics}, 225(1), 88--106.

\bibitem[Card(1995)]{card1995} Card, D. (1995). Using geographic variation in college proximity to estimate the return to schooling. In L.~N. Christofides, E.~K. Grant, \& R. Swidinsky (Eds.), \textit{Aspects of Labour Market Behaviour: Essays in Honour of John Vanderkamp} (pp. 201--222). University of Toronto Press.

\bibitem[Conley et al.(2012)]{conley2012} Conley, T.~G., Hansen, C.~B., \& Rossi, P.~E. (2012). Plausibly exogenous. \textit{Review of Economics and Statistics}, 94(1), 260--272.

\bibitem[Cragg(1997)]{cragg1997} Cragg, J.~G. (1997). Using higher moments to estimate the simple errors-in-variables model. \textit{RAND Journal of Economics}, 28, S71--S91.

\bibitem[Dagenais \& Dagenais(1997)]{dagenais1997} Dagenais, M.~G., \& Dagenais, D.~L. (1997). Higher moment estimators for linear regression models with errors in the variables. \textit{Journal of Econometrics}, 76(1--2), 193--221.

\bibitem[Dieterle \& Snell(2016)]{dieterle2016} Dieterle, S.~G., \& Snell, A. (2016). A simple diagnostic to investigate instrument validity and heterogeneous effects when using a single instrument. \textit{Labour Economics}, 42, 76--86.

\bibitem[Erickson \& Whited(2002)]{erickson2002} Erickson, T., \& Whited, T.~M. (2002). Two-step GMM estimation of the errors-in-variables model using high-order moments. \textit{Econometric Theory}, 18(3), 776--799.

\bibitem[Eriksson \& Koivunen(2004)]{eriksson2004} Eriksson, J., \& Koivunen, V. (2004). Identifiability, separability, and uniqueness of linear ICA models. \textit{IEEE Signal Processing Letters}, 11(7), 601--604.

\bibitem[Geary(1942)]{geary1942} Geary, R.~C. (1942). Inherent relations between random variables. \textit{Proceedings of the Royal Irish Academy, Section A}, 47, 63--76.

\bibitem[Gouri\'eroux et al.(2017)]{gouriéroux2017} Gouri\'eroux, C., Monfort, A., \& Renne, J.-P. (2017). Statistical inference for independent component analysis: Application to structural VAR models. \textit{Journal of Econometrics}, 196(1), 111--126.

\bibitem[Guay(2021)]{guay2021} Guay, A. (2021). Identification of structural vector autoregressions through higher unconditional moments. \textit{Journal of Econometrics}, 225(1), 27--46.

\bibitem[Hall(1992)]{hall1992} Hall, P. (1992). \textit{The Bootstrap and Edgeworth Expansion}. Springer.

\bibitem[Hoyer et al.(2008)]{hoyer2008} Hoyer, P.~O., Shimizu, S., Kerminen, A.~J., \& Palviainen, M. (2008). Estimation of causal effects using linear non-Gaussian causal models with hidden variables. \textit{International Journal of Approximate Reasoning}, 49(2), 362--378.

\bibitem[Huber \& Mellace(2015)]{huber2015} Huber, M., \& Mellace, G. (2015). Testing instrument validity for LATE identification based on inequality moment constraints. \textit{Review of Economics and Statistics}, 97(2), 398--411.

\bibitem[Kitagawa(2015)]{kitagawa2015} Kitagawa, T. (2015). A test for instrument validity. \textit{Econometrica}, 83(5), 2043--2063.

\bibitem[Lanne et al.(2017)]{lanne2017} Lanne, M., Meitz, M., \& Saikkonen, P. (2017). Identification and estimation of non-Gaussian structural vector autoregressions. \textit{Journal of Econometrics}, 196(2), 288--304.

\bibitem[Lanne et al.(2026)]{lanne2026} Lanne, M., Liu, K., \& Luoto, J. (2026). Identifying structural vector autoregressions via non-Gaussianity of potentially dependent shocks. \textit{The Econometrics Journal}, forthcoming.

\bibitem[Lewbel(1997)]{lewbel1997} Lewbel, A. (1997). Constructing instruments for regressions with measurement error when no additional data are available, with an application to patents and R\&D. \textit{Econometrica}, 65(5), 1201--1213.

\bibitem[Mourifi\'e \& Wan(2017)]{mourifie2017} Mourifi\'e, I., \& Wan, Y. (2017). Testing local average treatment effect assumptions. \textit{Review of Economics and Statistics}, 99(2), 305--313.

\bibitem[Nevo \& Rosen(2012)]{nevo2012} Nevo, A., \& Rosen, A.~M. (2012). Identification with imperfect instruments. \textit{Review of Economics and Statistics}, 94(3), 659--671.

\bibitem[Reiers\o l(1950)]{reiersol1950} Reiers\o l, O. (1950). Identifiability of a linear relation between variables which are subject to error. \textit{Econometrica}, 18(4), 375--389.

\bibitem[Rigobon(2003)]{rigobon2003} Rigobon, R. (2003). Identification through heteroskedasticity. \textit{Review of Economics and Statistics}, 85(4), 777--792.

\bibitem[Salehkaleybar et al.(2020)]{salehkaleybar2020} Salehkaleybar, S., Ghassami, A., Kiyavash, N., \& Zhang, K. (2020). Learning linear non-Gaussian causal models in the presence of latent variables. \textit{Journal of Machine Learning Research}, 21(39), 1--24.

\bibitem[Shimizu et al.(2006)]{shimizu2006} Shimizu, S., Hoyer, P.~O., Hyv\"arinen, A., \& Kerminen, A. (2006). A linear non-Gaussian acyclic model for causal discovery. \textit{Journal of Machine Learning Research}, 7, 2003--2030.

\bibitem[Shimizu et al.(2011)]{shimizu2011} Shimizu, S., Inazumi, T., Sogawa, Y., Hyv\"arinen, A., Kawahara, Y., Washio, T., Hoyer, P.~O., \& Bollen, K. (2011). DirectLiNGAM: A direct method for learning a linear non-Gaussian structural equation model. \textit{Journal of Machine Learning Research}, 12, 1225--1248.

\bibitem[Silva \& Shimizu(2017)]{silva2017} Silva, R., \& Shimizu, S. (2017). Learning instrumental variables with structural and non-Gaussianity assumptions. \textit{Journal of Machine Learning Research}, 18(120), 1--49.

\bibitem[Wang \& Drton(2023)]{wang2023} Wang, Y.~S., \& Drton, M. (2023). Causal discovery with unobserved confounding and non-Gaussian data. \textit{Journal of Machine Learning Research}, 24(165), 1--61.

\bibitem[Wooldridge(2010)]{wooldridge2010} Wooldridge, J.~M. (2010). \textit{Econometric Analysis of Cross Section and Panel Data} (2nd ed.). MIT Press.

\end{thebibliography}

\end{document}